\documentclass[runningheads,11points]{llncs}

\usepackage[russian,greek,english]{babel}
\usepackage[utf8x]{inputenc}
\usepackage{amsmath,ifthen,color,eurosym,tikz}

\usepackage{wrapfig,hyperref,cite}
\usepackage{latexsym,amsmath,amssymb,url}
\usepackage{fancyhdr}
\usepackage{graphicx}
\usepackage{stmaryrd,wasysym}
\usepackage{colortbl,datetime}

\usepackage[lined,boxed,commentsnumbered]{algorithm2e}

\newcommand{\remove}[1]{}

\spnewtheorem{observation}{Observation}{\bfseries}{\itshape}
\spnewtheorem{claimN}{Claim}{\bfseries}{\itshape}

  \usepackage{todonotes}

\newcommand{\tw}{\mathbf{tw}}

\newcommand{\tcw}{\mbox{\bf tcw}}

\newcommand{\minbi}{{\sc Min Bisection}}
\newcommand{\intcw}{\mathbf{in}\text{-}\mathbf{tcw}}

\renewenvironment{proof}{\noindent \textsc{Proof:}}{\hfill$\square$\medskip}
\newenvironment{proof-sketch}{\noindent \textsc{Sketch of proof:}}{\hfill$\square$\medskip}

\newenvironment{proofof}{\noindent \textsc{Proof of the Claim:}}{\hfill$\Diamond$\medskip}
\title{An FPT 2-Approximation for \\ Tree-Cut Decomposition\thanks{The research of the last author was
co-financed by the European Union (European Social Fund ESF) and Greek national
funds through the Operational Program ``Education and Lifelong Learning'' of the
National Strategic Reference Framework (NSRF), Research Funding Program: ARISTEIA~II.
The second author was supported by Basic Science Research
  Program through the National Research Foundation of Korea (NRF)
  funded by  the Ministry of Science, ICT \& Future Planning
  (2011-0011653).} \thanks{An extended abstract of this article will appear in the \emph{Proceedings of the 13th Workshop on Approximation and Online Algorithms} (WAOA), Patras, Greece, September 2015.}}
\author{Eunjung Kim~\inst{1}, Sang-il Oum~\inst{2}, Christophe
Paul~\inst{3}, Ignasi Sau~\inst{3}, and\\ Dimitrios M. Thilikos~\inst{3,4,5}}

\authorrunning{E. Kim, S. Oum, C. Paul, I. Sau, and D.~M. Thilikos}
\titlerunning{An FPT 2-Approximation for Tree-Cut Decomposition}

\institute{CNRS, LAMSADE, Paris, France.\\
\email{eunjungkim78@gmail.com}
\and Department of Mathematical Sciences, KAIST,  Daejeon, South Korea.\\
\email{sangil@kaist.edu}
\and CNRS, Université de Montpellier, LIRMM, Montpellier, France.\\
\email{christophe.paul@lirmm.fr},\ \email{ignasi.sau@lirmm.fr},\ \email{sedthilk@thilikos.info}
\and Department of Mathematics, University of Athens, Greece.
\and Computer Technology Institute  Press  ``Diophantus'',  Patras, Greece.}

\begin{document}

\maketitle
\setcounter{footnote}{0}

\begin{abstract}
The tree-cut width of a graph is a graph parameter
defined by Wollan~\textsl{[J. Comb. Theory, Ser. B, 110:47--66, 2015]}
with the help of tree-cut decompositions.
In certain cases, tree-cut width appears to be more adequate than treewidth
as an invariant that, when bounded, can accelerate the resolution of
intractable problems. While designing  algorithms
for problems with bounded tree-cut width, it is
important to have a parametrically tractable way to
compute the exact value of this parameter or, at least,
some constant approximation of it. In this paper we give
a parameterized $2$-approximation algorithm for the computation
of tree-cut width; for an input $n$-vertex graph $G$
and an integer $w$, our algorithm either confirms that the tree-cut width of $G$
is more than $w$ or returns  a tree-cut decomposition
of $G$ certifying that its tree-cut width is at most $2w$, in time $2^{O(w^2\log w)} \cdot n^2$. Prior to this work, no {\sl constructive} parameterized algorithms, even approximated ones, existed  for computing the tree-cut width of a graph. As a consequence of the Graph Minors series by Robertson and Seymour, only the {\sl existence} of a decision algorithm was known.
\end{abstract}

\noindent\textbf{Keywords:} Fixed-Parameter Tractable algorithm; tree-cut width; approximation algorithm.

\section{Introduction}
\label{sec:intro}

One of the most popular ways to decompose a graph into smaller pieces is given by the notion
of a tree decomposition. Intuitively, a  graph $G$ has a tree decomposition of small width
if it can be decomposed into small (possibly overlapping) pieces that are altogether arranged in a tree-like structure.
The {\em width} of such a decomposition is defined as the minimum size of these pieces.
The graph invariant of {\em treewidth} corresponds to the minimum width of  all
possible tree decompositions
and, that way, serves as a measure of the topological resemblance of a graph to the structure of
a tree.
The importance of tree decompositions and treewidth
in graph algorithms resides in the fact that
a wide family of  ${\sf {\sf NP}}$-hard graph problems
admits {\sf FPT}-algorithms, i.e., algorithms that run in $f(w)\cdot n^{O(1)}$
steps, when parameterized by the treewidth $w$ of their input graph.
According to the celebrated theorem of Courcelle, for every problem that
can be expressed in Monadic Second Order Logic (MSOL)~\cite{Courcelle90}
it is possible to design
an $f(w)\cdot n$-step algorithm on graphs of treewidth at most $w$.
Moreover, towards improving the parametric dependence, i.e., the
function $f$, of this algorithm for specific problems,
it is possible to design tailor-made dynamic programming
algorithms on the corresponding tree decompositions.
Treewidth has also been important from the combinatorial point of view.
This is mostly due to the celebrated ``{\sl planar
graph  exclusion theorem}''~\cite{RobertsonS-III,RobertsonS-V}. This theorem asserts that:

\begin{quote}(*) {\sl Every
graph that does not contain some fixed wall\footnote{We avoid the formal definition of a wall here. Instead, we provide the following image
\scalebox{0.05}{\includegraphics{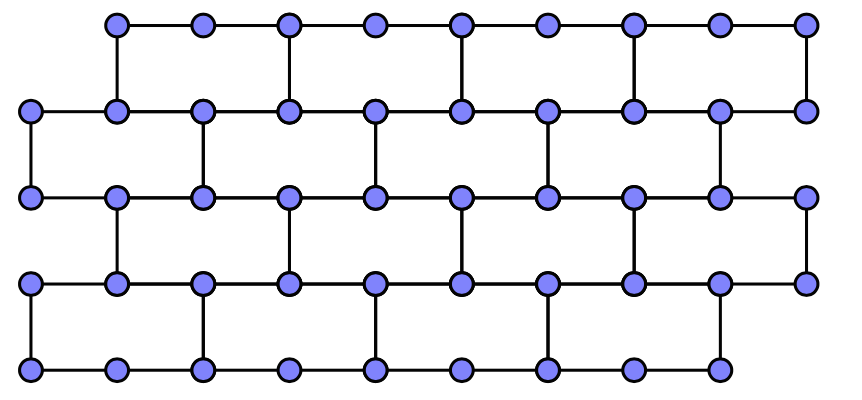}}
that, we believe, provides the necessary intuition.} as a topological minor\footnote{A graph $H$ is a {\em topological minor} of
a graph $G$ if  a subdivision of $H$ is a subgraph of $G$.} has bounded treewidth.} \smallskip

\end{quote}

 The above  result
had a considerable algorithmic impact as every problem for which a negative (or positive)
answer can be certified by the existence of some sufficiently big wall in its input,
is reduced to its resolution on graphs of bounded treewidth.
This induced a lot of research on the derivation of
fast parameterized algorithms that can construct (optimally or approximately)
these decompositions. For instance, according to~\cite{Bod96alin},
treewidth can be computed in $f(OPT)\cdot n$
steps where $f(w)=2^{O(w^3)}$ while, more recently, a 5-approximation
for treewidth was given in~\cite{BodlaenderDDFLP13anap} that runs in $2^{O(OPT)}\cdot n$ steps.

Unfortunately, the aforementioned success stories about treewidth
have some natural limitations. In fact, it is not always
possible to use treewidth for improving the tractability of {\sf {\sf NP}}-hard problems.
In particular, there are interesting cases of problems where no such
an {\sf FPT}-algorithm is expected to
exist~\cite{FellowsFLRSST11,GolovachT11path,DomLSV08capa}.
Therefore,
it is an interesting question whether there are alternative, but still general,
graph invariants
that can provide tractable parameterizations
for such problems.


A promising candidate in this direction
is the graph invariant of {\sl tree-cut width}
that was recently
introduced by Wollan in~\cite{Wollan15thes}. Tree-cut width can be seen as  an ``edge'' analogue of  treewidth. 
It is defined using a different
type of decompositions, namely, tree-cut decompositions that
are roughly tree-like partitions of a graph into mutually disjoint pieces
such that both the size of some ``essential'' extension of these pieces and the
number of edges crossing two neighboring pieces are bounded
(see Section~\ref{sec:prelim} for the formal definition).
Our first result is that it is {\sf {\sf NP}}-hard to decide, given a graph $G$ and an integer $w$,
whether the input graph $G$ has tree-cut width at most $w$. This follows from a reduction from the \minbi\ problem that is presented in Subsection~\ref{npbpdp}. This encourages us to consider a parameterized algorithm
for this problem.



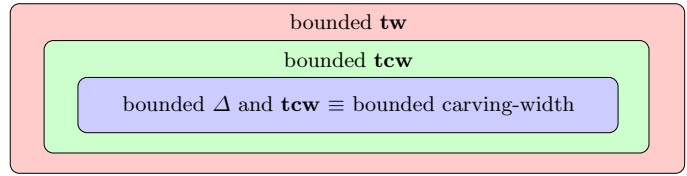
\begin{figure}
\begin{center}
 \scalebox{.9}{
\begin{tikzpicture}
   \draw[rounded corners=1ex,fill=red!20] (-1,-.6) rectangle (255pt,14ex);
   \draw[rounded corners=1ex,fill=green!20] (-.5,-.3) rectangle (240pt,10ex);
   \draw[rounded corners=1ex,fill=blue!20] (0,0) rectangle (227pt,6ex);
     \node[] at (4,.4) {bounded $\Delta$ and ${\bf tcw}$ $\equiv$ bounded carving-width  };
       \node[] at (4,1.08) {bounded ${\bf tcw}$  };
       \node[] at (4,1.65) {bounded ${\bf tw}$  };
\end{tikzpicture}
}
\end{center}
\caption{The relations between classes with bounded treewidth (${\bf tw}$) and tree-cut width (${\bf tcw}$).}
\label{fig:incl}
\end{figure}

Another tree-like parameter that can be seen as an edge-counterpart of tree\-width
is {\sl carving-width}, defined in~\cite{SeymourT94call}. It is known that a graph has
bounded carving-width if and only if both its treewidth and its maximum degree are bounded. 
We stress that
this is not the case for tree-cut width, which can also  capture  graphs with unbounded maximum degree and, thus, is more general than carving-width.
There are two
reasons why tree-cut width might be a good alternative for treewidth. We expose them below.

\medskip

 \noindent {\bf (1) Tree-cut width as a parameter.}  From now on we denote by $\tcw(G)$ (resp.  $\tw(G)$) the tree-cut width (resp. treewidth) of a graph $G$.
 As it is shown in~\cite{Wollan15thes}  $\tcw(G)=O(\tw(G)\cdot \Delta(G))$.
Moreover, in~\cite{GKS14}, it was proven that $\tw(G)=O((\tcw(G))^2)$ and in Subsection~\ref{vsvs8plq}, we prove that the latter upper bound
is asymptotically tight.
The graph class inclusions generated by  the aforementioned relations are
depicted in Fig.~\ref{fig:incl}. As tree-cut width is a ``larger'' parameter than treewidth, one may expect that some problems that are intractable when parameterized by treewidth (known to be {\sf W}$[1]$-hard or open) become tractable when parameterized by tree-cut width. Indeed, some recent  progress
on the  development of a  dynamic programming framework for
tree-cut width (see~\cite{GKS14}) confirms that assumption.
According to~\cite{GKS14}, such problems include {\sc Capacitated
Dominating Set} problem,
 {\sc Capacitated Vertex Cover}~\cite{DomLSV08capa}, and {\sc Balanced Vertex-Ordering} problem. We expect that more problems will fall into this category.
 \medskip\medskip

\noindent{\bf (2) Combinatorics of tree-cut width.}
In~\cite{Wollan15thes} Wollan proved the following counterpart of (*):
\begin{quote}\smallskip
(**) {\sl Every
graph that does not contain some fixed wall as an immersion\footnote{A graph $H$ is an {\em immersion} of
a graph $G$ if $H$ can be obtained from some subgraph of $G$ after replacing edge-disjoint paths with edges.} has bounded tree-cut width.} \smallskip
\end{quote}
Notice that (*) yields (**) if we replace ``topological minor''  by ``immersion'' and
``treewidth'' by ``tree-cut width''. This implies that tree-cut width has
combinatorial properties analogous to those of treewidth.
It  follows that every problem where a negative (or positive) answer
can be certified by the existence of a wall as an immersion, can be reduced
to the design of a suitable dynamic programming algorithm for this problem 
     on graphs of bounded tree-cut width.\medskip

\paragraph{\bf Computing tree-cut width.} It follows that designing dynamic programming algorithms
on tree-cut decompositions might  be a promising task when this is not possible (or promising)
on tree-decompositions. Clearly, this makes it imperative to have an efficient algorithm
that, given a graph $G$ and an integer $w$,
 constructs tree-cut decompositions of width at most $w$ or reports that this is not possible.
Interestingly, an $f(w)\cdot n^3$-time algorithm for the {\em decision version} of the problem is known to {\em exist} but this is not done in a constructive way. Indeed, for every fixed $w$, the class of graphs with tree-cut width at most $w$ is
closed under immersions~\cite{Wollan15thes}. By the fact that
graphs are well-quasi-ordered under immersions~\cite{RobertsonS10-XXIII},
for every $w$, there exists a {\sl finite}
 set ${\cal R}_{w}$ of graphs such that $G$ has tree-cut width at most $w$
if and only if it does not contain any of the graphs in ${\cal R}_{w}$ as an immersion.  From~\cite{GroheKMW11find}, checking whether an $h$-vertex
graph $H$ is contained as an immersion in some $n$-vertex graph $G$
can be done in $f(w)\cdot n^{3}$ steps. It follows that, for every fixed $w$, there {\sl exists} a polynomial algorithm checking whether the tree-cut width of a graph
is at most $w$. Unfortunately, the  {\sl construction} of this algorithm
requires the knowledge of  the set ${\cal R}_{w}$ for every $w$, which is
not provided by the results in~\cite{RobertsonS10-XXIII}. Even if we knew ${\cal R}_w$, it is not clear how to construct a tree-cut decomposition of width at most $w$, if one exists.

In this paper we make a first step
towards a constructive parameterized algorithm for tree-cut width by giving
an {\sf FPT} 2-approximation for it. Given a graph $G$ and an integer $w$, our algorithm either reports that $G$ has tree-cut width more than $w$ or outputs a tree-cut decomposition of width at most $2w$  in $2^{O(w^2\log w)}n^{2}$ steps. The algorithm is presented in Section~\ref{sec:algorithm}.

\section{Problem definition and preliminary results}
\label{sec:prelim}

Unless specified otherwise, every graph in this paper is undirected and loopless and may have multiple edges. By $V(G)$ and $E(G)$ we denote the vertex set and the edge set, respectively, of a graph $G$. Given a vertex $x\in V(G)$, the {\em neighborhood} of $x$ is $N(x)=\{y\in V(G)\mid xy\in E(G)\}$. Given two disjoint sets $X$ and $Y$ of $V(G)$, we denote $\delta_G(X,Y)=\{xy\in E(G)\mid x\in X, y\in Y\}$. For a subset $X$ of $V(G)$, we define $\partial_G(X)=\{x\in X \mid N(x)\setminus X\neq\emptyset\}$.

\subsection{Tree-cut width and treewidth}

\noindent{\bf Tree-cut width}.
A {\em tree-cut decomposition} of $G$ is a pair $(T,\cal{X})$ where $T$ is a tree and ${\cal X}=\{X_t\subseteq V(G)\mid t\in V(T)\}$ such that \begin{itemize}
\item[$\bullet$] $X_t\cap X_{t'}=\emptyset$ for all distinct $t$ and $t'$ in $V(T)$,
\item[$\bullet$] $\bigcup_{t\in V(T)} X_t=V(G)$.
\end{itemize}
From now on we refer to the vertices of $T$ as {\em nodes}. The sets in $\cal{X}$ are called the {\em bags} of the tree-cut decomposition. Observe that the conditions above allow to assign an empty bag for some node of $T$. Such nodes are called {\em trivial nodes}. Observe that we can always assume that trivial nodes are internal nodes.

Let $L(T)$ be the set of leaf nodes of $T$. For every tree-edge $e=\{u,v\}$ of $E(T)$, we let $T_u$ and $T_v$ be the subtrees of $T\setminus e$ which contain $u$ and $v$, respectively.

We define the \emph{adhesion} of a tree-edge $e=\{u,v\}$ of $T$ as follows:
$$\delta^T(e)=\delta_{G}(\bigcup_{t\in V(T_u)} X_t, \bigcup_{t\in V(T_v)} X_t).$$

For a graph $G$ and a set $X\subseteq V(G)$, the {\em 3-center} of $(G,X)$ is the graph obtained from $G$ by repetitively
dissolving every vertex $v \in V(G)\setminus X$ that has two neighbors and degree 2 and removing every
vertex $w \in V(G)\setminus X$ that has degree at most 2 and one neighbor ({\em dissolving} a vertex
$x$ of degree two with exactly two neighbors $y$ and $z$ is the operation of removing $x$ and adding
the edge $\{y,z\}$ -- if this edge already exists then its multiplicity is increased by one).

Given a tree-cut decomposition $(T,\cal{X})$ of $G$ and node $t\in V(T)$, let $T_1, \ldots, T_{\ell}$
be the connected components of $T\setminus t$. The {\em torso} of $G$ {\em at} $t$, denoted by $H_t$, is a graph obtained
from $G$ by identifying each non-empty vertex set $Z_i:=\bigcup_{b\in V(T_i)} X_{b}$ into a single
vertex $z_i$ (in this process, parallel edges are kept). We denote by $\bar{H}_t$
the 3-center of $(H_t,X_t)$. Then the {\em  width} of $(T,\cal{X})$ equals
$$\max\ (\{|\delta^T(e)|: e\in E(T)\}\ \cup\  \{|V(\bar{H}_t)|: t\in V(T)\}).$$
The \emph{tree-cut width} of $G$, or $\textbf {tcw}(G)$ in short, is
the minimum width of $(T,\cal{X})$ over all tree-cut decompositions $(T,\cal{X})$ of $G$.

\medskip
The following definitions will be used in the approximation algorithm.
%
Let $(T,{\cal X})$ be a tree-cut decomposition of $G$. It is \emph{non-trivial} if it contains at least two non-empty bags, and {\em trivial} otherwise. We will assume that every leaf of a tree-cut decomposition has a non-empty bag. The {\em internal-width} of a non-trivial tree-cut decomposition $(T,{\cal X})$ is
$$\intcw(T,\mathcal{X})=\max\ (\{|\delta^T(e)|: e\in E(T)\}\ \cup\  \{|V(\bar{H}_t)|: t\in V(T)\setminus L(T)\}).$$
If $(T,{\cal X})$ is trivial, then we set $\intcw(T,\mathcal{X})=0$.
\medskip

We decision problem corresponding to tree-cut width is the following:

\begin{center}
\fbox{\small\begin{minipage}{11,9cm}
\noindent{\sc Tree-cut Width}\\
{\sl Input}: a plane graph $G$ and a non-negative integer $k$.\\
{\sl Question}: $\tcw(G)\leq k$?
\end{minipage}}
\end{center}

%

\medskip
\noindent
\textbf{Treewidth}.
A \emph{tree decomposition} of a graph $G$ is a pair $(T,{\cal Y}) = \{Y_x : x \in V(T) \})$ such that $T$ is a tree and  ${\cal Y}$ is a collection of subsets of $V(G)$ where
\begin{itemize}
\item[$\bullet$] $\bigcup_{x \in V(T)} Y_x = V(G)$;
\item[$\bullet$] for every edge $\{u,v\} \in E (G)$ there exists $x \in V(T)$ such that $u,v \in Y_x$; and
\item[$\bullet$] for every vertex $u \in V(G)$ the set of nodes $\{ x\in V(T) : u \in Y_x \}$ induces a subtree of $T$.
\end{itemize}
The vertices of $T$ are called {\em nodes} of $(T,{\cal Y})$  and the sets $Y_x$ are called bags.
The \emph{width} of a tree decomposition is the size of the largest bag minus one. The \emph{treewidth} of a graph, denoted by $\tw(G)$, is the smallest width of a tree decomposition of $G$.

\subsection{Computing tree-cut width is {\sf NP}-complete}
\label{npbpdp}

We prove that {\sc Tree-cut Width} is {\sf NP}-hard by a polynomial-time reduction from \minbi, which is known to be {\sf NP}-hard~\cite{GareyJ79}. The input of  \minbi\ is a graph $G$ and a non-negative integer $k$, and
the question is whether there exists a bipartition $(V_1,V_2)$ of $V(G)$ such that $|V_1|=|V_2|$ and $|\delta_G(V_1,V_2)|\leqslant k$.

\begin{theorem}\label{th:npc}
{\sc Tree-cut Width} is {\sf NP}-complete.
\end{theorem}
\begin{proof}
It is easy to see that \textsc{Tree-cut Width} is in {\sf NP}. We present a reduction from \minbi\ to {\sc Tree-cut Width}   (see Fig.~\ref{fig:tcw-npc}).
Let $(G,k)$ be an instance of \minbi\ on $n$ vertices. We may assume that $k\leqslant n^2$ since otherwise, the instance is trivially NO. We create an instance $(G',w)$  with $w=\frac{n^3}{2}+k$ as follows. The vertex set $V(G')$ consists of a set $V$ of size $n$, a set $Q$ of size $w-2$, and the set $C_{x,y}$ of size $w+1$ for every pair $x,y\in Q$. Edges are added so that:
\begin{itemize}
\item[$\bullet$] $G'[V]=G$.
\item[$\bullet$] For every pair $x,y \in Q$, all vertices of $C_{x,y}$ are adjacent with both $x$ and $y$.
\item[$\bullet$] Each $x\in V$ is adjacent with $n^2$ (arbitrarily chosen) vertices of $Q$.
\end{itemize}
We now proceed with the proof of the correctness of the above reduction.
Suppose that $(G,k)$ is a {\sc Yes}-instance to {\sc Min Bisection} with a bipartition $(V_1,V_2)$. Consider a tree-cut decomposition $(T,{\cal X})$ in which $V(T)$ contains three nodes $t_1,t_2,q$ and some additional nodes. The tree $T$ forms a star with $q$ as the center and all other nodes as leaves. We have $X_{t_i}=V_i$ for $i=1,2$, $X_q=Q$ and each vertex of $\bigcup_{x,y \in Q}C_{x,y}$ forms a singleton bag. It is not difficult to verify that $(T,{\cal X})$ is a tree-cut decomposition of $G'$ whose width is $w$. In particular, notice that $|V(\bar{H}_q)|=|Q|+2=w$ and $|\delta(t_i,q)|=\frac{n}{2}\cdot n^2+k=w$ for $i=1,2$.

Conversely, suppose that $G'$ admits a tree-cut decomposition $(T,{\cal X})$ of width at most $w$. Any two vertices $x,y\in Q$ must be in the same bag since they are connected by $w+1$ disjoint paths via $C_{x,y}$. Hence, there exists a tree node, say $q$, in $T$ such that $Q\subseteq X_q$.

\begin{figure}[t]
\centerline{\scalebox{0.66}{\includegraphics{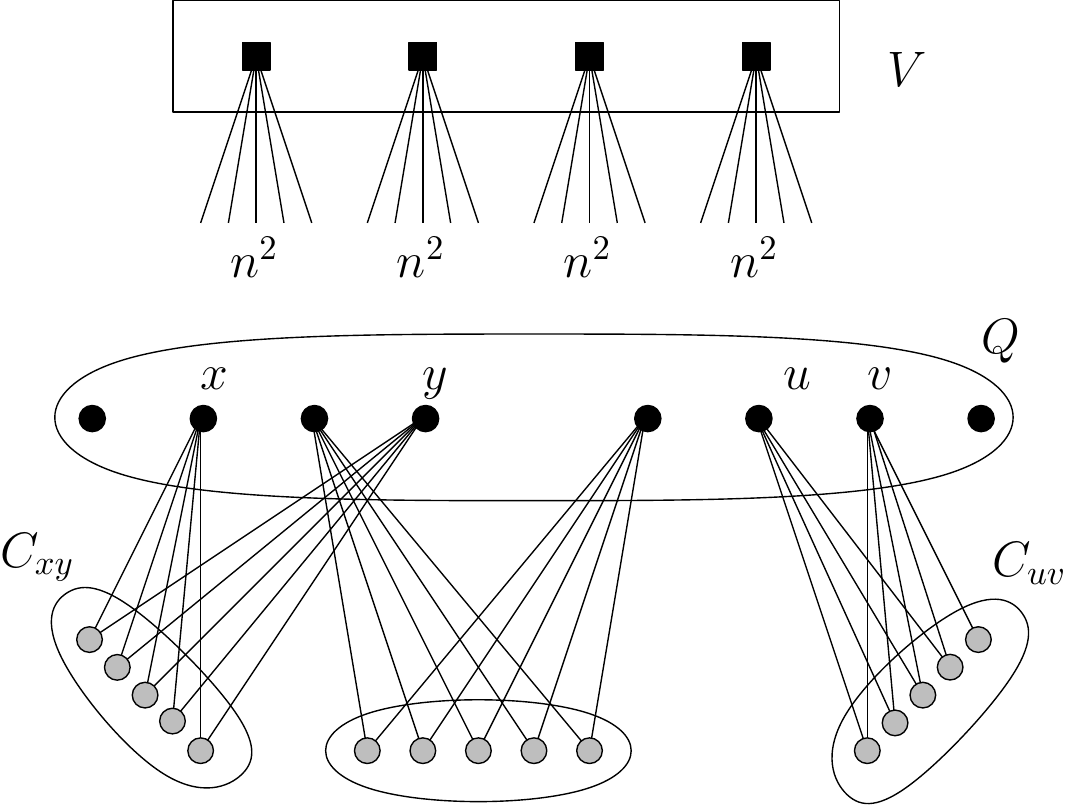}}}
\caption{The graph $G'$ in the transformation of the instances of  \minbi\ to equivalent
instances of  {\sc Tree-cut Width}.}
\label{fig:tcw-npc}
\end{figure}

Consider the set ${\cal C}=\{T_1,\ldots , T_{\ell}\}$ of the connected components of $T\setminus \{q\}$ and let $e_i$ be the tree-edge between $T_i$ and $q$. As $w\geqslant |V(\bar{H}_q)|\geqslant |Q|=w-2$,  there are at most two tree-edges among $e_1,\ldots , e_{\ell}$ such that $|\delta^{T}(e_i)|\geqslant 3$. This means that there are at most two subtrees among $T_1,\ldots , T_{\ell}$ such that $V\cap \bigcup_{t\in V(T_i)}X_t\neq \emptyset$. From the fact that $|Q|=w-2$, at least $n-2$ vertices of $V$ are {\em not} contained
in $X_q$ and thus there exists at least
one subtree $T_i$ such that $V\cap \bigcup_{t\in V(T_i)}X_t\neq \emptyset$. If
there  is $i$ such that $|V\cap \bigcup_{t\in V(T_i)}X_t|\geqslant \frac{n}{2}+1$,
then $|\delta^{T}(e_i)|\geqslant (\frac{n}{2}+1)\cdot n^2>w$, a contradiction.
Hence, we conclude
that there are exactly two
subtrees, say $T_1$ and $T_2$, in ${\cal C}$ such that $V\cap \bigcup_{t\in V(T_i)}X_t\neq \emptyset$ for $i=1,2$ and for $3\leqslant i\leqslant \ell$, we have $V\cap \bigcup_{t\in V(T_i)}X_t= \emptyset$. This, together with the fact that $|Q|=w-2$, enforces that the sets $V\cap \bigcup_{t\in V(T_1)}X_t$ and $V\cap \bigcup_{t\in V(T_2)}X_t$  make a bipartition of $V$ into sets of equal size. Let us call this bipartition $\{V_1,V_2\}$.
Observe that $\delta^{T}(e_i)\supseteq \delta(V_i,Q)\cup \delta(V_1,V_2)$,
thus $\delta^{T}(e_i)$ contains  at least $\frac{n}{2}\cdot n^2+|\delta(V_1,V_2)|$
edges for $i=1,2$. As $|\delta^{T}(e_1)|\leqslant w$, it
follows $|\delta(V_1,V_2)|\leqslant k$. Therefore, $(G,k)$ is {\sc Yes}-instance to \minbi\, which complete the proof.
\end{proof}

\subsection{Tree-cut width vs treewidth}
\label{vsvs8plq}

In this section we investigate  the relation between treewidth and tree-cut width.
The following was proved in~\cite{GKS14}.

\begin{proposition}
\label{tfdwtcw}
For a graph of tree-cut width at most $w$, its treewidth is at most $2w^2+3w$.
\end{proposition}

In the rest of this subsection we prove
that the bound of Proposition~\ref{tfdwtcw} is asymptotically optimal. For this we need some definitions.\\

Let $G$ be a graph.
Two subgraphs $X$ and $Y$ of $G$ {\em touch} each other if either $V(X)\cap V(Y)\neq \emptyset$ or there is an edge $e=\{x,y\} \in E(G)$ with $x\in V(X)$ and $y\in V(Y)$. A {\em bramble} ${\cal B}$ is a collection of connected subgraphs of $G$ pairwise touching each other. The {\em order} of a bramble ${\cal B}$ is the minimum size of a hitting set $S$ of ${\cal B}$, that is  a set $S\subseteq V(G)$ such that for every $B\in {\cal B}$, $S\cap V(B)\neq \emptyset$. In Seymour and Thomas~\cite{SeymourT93}, it is known that the treewidth of a graph equals the maximum order over all brambles of $G$ minus one. Therefore, a bramble of  order $k$ is a certificate that the treewidth is at least $k-1$. \smallskip

We next define the  family of graphs ${\cal H}=\{H_w:w\in {\mathbb N}_{\geqslant 1}\}$ as follows.
 The vertex set of $H_w$ is a disjoint union of $w$ cliques, $Q_1,\ldots , Q_w$, each containing  $w$ vertices. For each $1\leqslant i\leqslant w$, the vertices of $Q_i$ are labeled as $(i,j)$, $1\leqslant j\leqslant w$. Besides the edges lying inside the cliques $Q_i$'s, we add an edge between $(i,j)\in Q_i$ and $(j,i)\in Q_j$ for every $1\leqslant i < j\leqslant w$. Notice that the vertex $(i,i)$ does not have a neighbor outside $Q_i$. The graph $H_4$ is depicted in Fig.~\ref{fig:h4}.

\begin{figure}[h]
\centerline{\scalebox{0.66}{\includegraphics{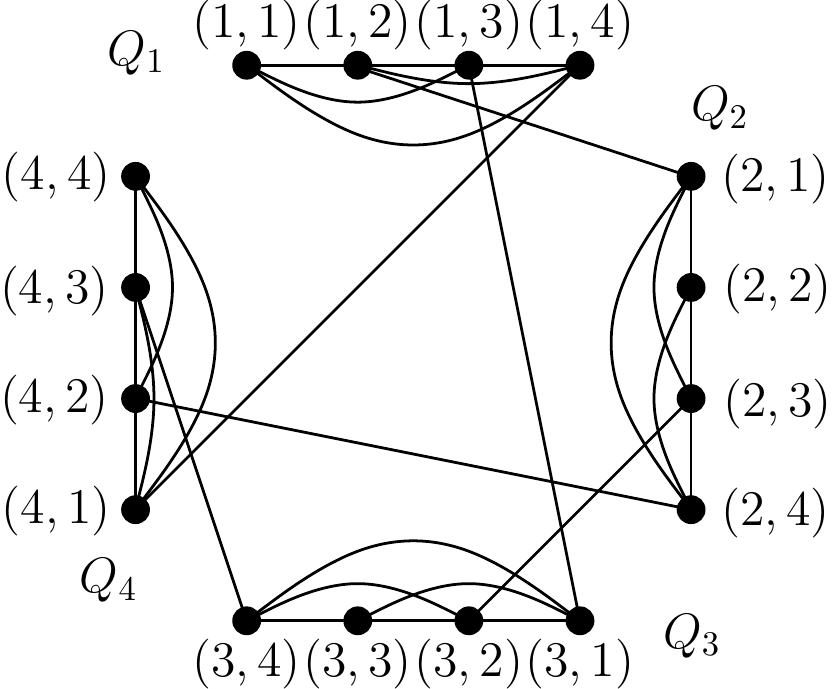}}}
\caption{The graph $H_4$.}
\label{fig:h4}
\end{figure}

\begin{lemma}
\label{rji3rlp}
The tree-cut width of $H_w$ is at most $w+1$.
\end{lemma}

\begin{proof}
Consider the tree-cut decomposition $(T,{\cal X})$, in which $T$ is a star with $t$ as the center and $q_1,\ldots , q_w$ as leaves. For the bags, we set $X_t=\emptyset$, and $X_{q_i}=Q_i$ for $1\leqslant i\leqslant w$. It is straightforward to verify that the tree-cut width of $(T,{\cal X})$ is $w+1$.
\end{proof}

\begin{lemma}
\label{l9l9f3j5}
For any positive integer $w$, the treewidth of $H_w\in {\cal H}$ is at least $\frac{1}{16}w^2-1$.
\end{lemma}
\begin{proof}
For notational convenience, we assume that $w$ is even. The argument can be easily extended to the case when $w$ is odd. For $i\in [w]$ and a set $Z\subseteq [w]$, let $B(i,Z)$ denote the set $\{(i,j),(j,i):j\in Z\}$. We define ${\cal B}_w$ as \[{\cal B}_w=\{G[B(i,Z)]:\forall i\in [w], \ \forall Z\subseteq [w]\setminus \{i\} \text{~s.t.~} |Z|=w/2\}.\]

It is easy to verify that each subgraph of ${\cal B}_w$ is connected. For any $i\in [w]$ and $Z\subseteq [w]\setminus \{i\}$ such that $|Z|=\frac{1}{2}w$, the number of cliques $Q_i$, $1\leqslant i \leqslant w$, with which $B(i,Z)$ has non-empty intersection is at least $\frac{1}{2}w+1$. This means any two elements of ${\cal B}_w$ touch each other, and thus ${\cal B}_w$ is indeed a bramble. Henceforth, we show that the order of ${\cal B}_w$ is at least $\frac{1}{16}w^2$.

Suppose that there is a hitting set $S$ of ${\cal B}_w$ with $|S|<\frac{1}{16}w^2$. We define
$$F_{S}=\{i\in[w]\mid |\{j\in [w]: (j,i)\in V(G)\setminus S\}|\geqslant \frac{3}{4}w\}.$$

\begin{claimN}\label{Fsize}
$|F_{S}|>\frac{3}{4}w$.
\end{claimN}
\begin{proofof}
Suppose that the contrary holds. We use a counting argument to derive a contradiction. The set $V(G)\setminus S$ is partitioned into two sets: $\{(j,i): j\in [w],i\in F_{S}\}$ and $\{(j,i): j\in [w],i\not\in F_{S}\}$. We have
\[
|V(G)\setminus S| \leqslant w\cdot |F_{S}| + \frac{3}{4}w\cdot (w-|F_{S}|)
						\leqslant \frac{3}{4}w^2 + \frac{3}{16}w^2=\frac{15}{16}w^2,
\]
contradicting to the assumption that $|S|<\frac{1}{16}w^2$.
\end{proofof}

\begin{claimN}\label{bigrow}
There exists $i^*\in F_{S}$ such that $|\{j\in [w]:(i^*,j)\in  S\}|<\frac{1}{4}w$.
\end{claimN}
\begin{proofof}
Suppose the contrary, i.e. we have $|\{j\in [w]:(i,j)\in  V(G)\setminus S\}|\leqslant \frac{3}{4}w$ for every $i\in F_{S}$. Notice that the set $V(G)\setminus S$ is partitioned into $\{(i,j):i\in F_{S}, j\in [w]\}$ and $\{(i,j):i\in [w]\setminus F_{S}, j\in [w]\}$. Then,
\[|V(G)\setminus S| \leqslant |F_{S}|\cdot \frac{3}{4}w + (w-|F_{S}|)\cdot w \leqslant w^2 -\frac{1}{4}w\cdot |F_{S}| < \frac{13}{16}w^2, \]
where the last inequality follows from Claim~\ref{Fsize}. This contradicts the assumption that $|S|<\frac{1}{16}w^2$.
\end{proofof}

Consider some $i^*\in F_{S}$ satisfying the condition of Claim~\ref{bigrow}. We observe that the set \[Z=\{j\in [w]:(j,i^*)\in V(G)\setminus S\}\setminus (\{i^*\}\cup  \{j\in [w]:(i^*,j)\in S\})\]
contains at least $\frac{1}{2}w$ vertices by the definition of $F_{S}$ and Claim~\ref{bigrow}. Pick any subset $Z^*$ of $Z$ of size exactly $\frac{1}{2}w$. To reach a contradiction, it suffices to show that $B(i^*,Z^*)\cap S= \emptyset$. Indeed, from the fact that $Z^* \subseteq \{j\in [w]:(j,i^*)\in V(G)\setminus S\}$, it follows that
\begin{eqnarray}
 \forall j\in Z^*\  (j,i^*)\in V(G)\setminus S.\label{this}
 \end{eqnarray}
By the definition of $Z$ it follows that  $Z^*\cap \{j\in [w]:(i^*,j)\in S\}=\emptyset,$ which, implies that
 \begin{eqnarray}
 \forall j\in Z^*\  (i^*,j)\in V(G)\setminus S.\label{that}
 \end{eqnarray}
 By~\eqref{this} and~\eqref{that}, we conclude that $B(i^*,Z^*)\cap S= \emptyset$.
 This completes the proof. \end{proof}

From Lemmata~\ref{rji3rlp} and~\ref{l9l9f3j5}, we conclude to the following.

\begin{theorem}
\label{theokg8h}
For every $w\in {\mathbb N}_{\geqslant 1}$ there exists a graph $H_{w}$ such that
$\tw(H_{w})={\rm \Omega}((\tcw(H_{w}))^2)$.
\end{theorem}

\section{The 2-approximation algorithm}
\label{sec:algorithm}

We present a 2-approximation of {\sc Tree-cut Width} running in time $2^{O(w^2\log w)}\cdot n^2$.
As stated in Lemma~\ref{e92mfsd} below, we first observe that computing the tree-cut width of $G$ reduces to computing the tree-cut width of 3-edge-connected graphs. This property can be easily derived from~\cite[Lemmas 10--11]{Wollan15thes}.

\begin{lemma}\label{e92mfsd}
Given a connected graph $G$, let $\{V_1,V_2\}$ be a partition of $V(G)$ such that $\delta_{G}(V_1,V_2)$ is a minimal cut of size at most two and let $w\geqslant 2$ be a positive integer. For $i=1,2$, let $G_i$ be the graph obtained from $G$ by identifying the vertex set $V_{3-i}$ into a single vertex $v_{3-i}$. Then $G$ has tree-cut width at most $w$ if and only if both $G_1$ and $G_2$ have tree-cut width at most $w$.
\end{lemma}

\begin{proof}
Recall that $\textbf {tcw}(H)\leqslant \textbf{tcw}(G)$ if $G$ admits an immersion of $H$ by~\cite[Lemma 11]{Wollan15thes}. Hence, in order to prove the forward implication, it suffices to prove that $G_i$ is an immersion of $G$, for $i=1,2$. If $|\delta(V_1,V_2)|=1$, for each $i=1,2$ we can delete all vertices of $V_{3-i}$ except for the single vertex in $N(V_i)$ and obtained $G_i$. If $|\delta(V_1,V_2)|=2$, note that for each $i=1,2$, $G[V_{3-i}]$ is connected and thus $G_i$ can be obtained by deleting vertices, edges and lifting a sequence of edges along the path between two vertices in $N(V_i)$.

Conversely, let $(T^i,{\cal X}^i)$ be a tree-cut decomposition of $G_i$ of width at most $w$ for $i=1,2$, and consider the tree-cut decomposition $(T,{\cal X})$ such that ${\cal X}={\cal X}_1 \cup {\cal X}_2$ and $T$ is obtained by the disjoint union of $T^1$ and $T^2$ after adding an edge between $t_1\in V(T^1)$ and $t_2 \in V(T^2)$, where $t_i$ is the tree node of $T_i$ containing $v_{3-i}$, i.e. the vertex obtained by contracting $V_{3-i}$. We remove $v_1$ and $v_2$ from the bags of $T$.

We claim that $(T,{\cal X})$ is a tree-cut decomposition of width at most $w$. Note first that the adhesion of $(T,{\cal X})$ is at most $w$ since $|\delta^{T}(\{t_1,t_2\})|\leqslant 2$ and the adhesion of $(T^{i}, {\cal X}^i)$ is at most $w$ for $i=1,2$. From $|\delta^{T}(\{t_1,t_2\})|\leqslant 2$, it follows that for $i=1,2$, the 3-center of $(H_{t_i},X_{t_i})$ of the tree-decomposition $(T,{\cal X})$ is the same as the 3-center of $(H_{t_i},X_{t_i})$ of the tree-decomposition $(T^{i},{\cal X}^i)$. Therefore the width of $(T,{\cal X})$ is at most $w$.
\end{proof}

The proof of the next lemma is easy  and is omitted.

\begin{lemma}
\label{i8i30eieue89}
Let $G$ be a graph and let $v$ be a vertex of $G$ with degree 1 (resp. 2). Let also $G'$ be the graph
obtained from $G$ after removing (resp. dissolving) $v$. Then $\tcw(G)=\tcw(G')$.
\end{lemma}

From now on, based on Lemmata~\ref{e92mfsd} and~\ref{i8i30eieue89},
we assume that the input graph is $3$-edge-connected. In this special case, the following observation is not difficult to verify. It allows us to work with a slightly simplified definition of the $3$-centers in a tree-cut decomposition.

\begin{observation}\label{sldjflk34}
Let $G$ be a 3-edge-connected graph and let  $(T,\cal{X})$ be a
tree-cut decomposition of $G$. Consider an arbitrary node $t$ of $V(T)$ and let ${\cal T}$ be the set
containing every connected component $T'$ of $T\setminus t$ such that
$\bigcup_{s\in V(T')}X_{s}\neq \emptyset.$
Then $ |V(\bar{H}_t)|=|X_t|+|{\cal T}|,$ that is $|V(\bar{H}_{t})|=|V({H}_{t})|$.
\end{observation}

We observe that the proof of Lemma~\ref{e92mfsd} provides a way to construct a desired tree-cut decomposition for $G$ from decompositions of smaller graphs. Given an input graph $G$ for {\sc Tree-cut Width}, we find a minimal cut $(V_1,V_2)$ with $|\delta(V_1,V_2)|\leqslant 2$ and create a graph $G_i$ as in Lemma~\ref{e92mfsd}, with the vertex $v_{3-i}$ marked as distinguished. We recursively find such a minimal cut in the smaller graphs created until either one becomes 3-edge-connected or has at most $w$ vertices.

Therefore, a key feature of an algorithm for {\sc Tree-cut Width} lies in how to handle 3-edge-connected graphs. Our algorithm iteratively refines a tree-cut decomposition $(T,\mathcal{X})$ of the input graph $G$ and either guarantees that the following invariant is satisfied or returns that $\tcw(G)>\omega$.

\begin{quote}
\hspace{-.6cm}{\sl Invariant:} {\sl $(T,\mathcal{X})$ is a tree-cut decomposition of $G$ where $\intcw(T,\mathcal{X})\leqslant 2\cdot w$}.
\end{quote}

Clearly the trivial tree-cut decomposition satisfies  the {\sl Invariant}. A leaf $t$ of $T$ such that  $|X_t|\geqslant 2\cdot \omega$ is called a {\em large leaf}. At each step, the algorithm picks a large leaf  and refines
 the current tree-cut decomposition by breaking this leaf bag into smaller pieces. The process repeats until we
 finally obtain a tree-cut decomposition of width at most $2w$, or encounter a certificate that $\textbf {tcw}(G)>w$.



\subsection{Refining a large leaf of a tree-cut decomposition}
\label{sec:specialpart}

A large leaf will be further decomposed into a star. To that aim, we will solve the following problem:


\medskip
\noindent {\sc Constrained Star-Cut Decomposition}\\
\noindent {\sl Input:}  An 
undirected graph $G$, an integer $w\in\Bbb{N}$,  a set $B \subseteq V(G)$, and a weight function $\gamma: B\rightarrow \Bbb{N}$. 

\noindent {\sl Parameter:} $w$.\\
\noindent{\sl Output:}  A non-trivial tree-cut decomposition $(T,\mathcal{X})$ of $G$ such that
\begin{enumerate}
\item $T$ is a star with central node $t_c$ and
with $\ell$ leaves for some $\ell\in \Bbb{N}^+$,
\item $\intcw(T,\mathcal{X})\leqslant w$, and
\item $|X_{t_c}|+\ell\leqslant w$ and for every leaf node $t$, $\gamma(B\cap X_t) \leqslant w$,
\end{enumerate}
or report that such a tree-cut decomposition does not exist.

\medskip
Observe that a {\sc Yes}-instance satisfies, for every $x\in B$, $\gamma(x)\leqslant w$.
We also notice that as the output of the algorithm is a non-trivial tree-cut decomposition, $T$ contains at least two nodes with non-empty bags and every leaf node is non-empty.

Given a subset $S\subseteq V(G)$, we  define the instance of the \textsc{Constrained Star-Cut
Decomposition} problem $I(S,G)=(G[S],w,\partial_G(S),\gamma_S)$ where
for every $x\in \partial_G(S)$,
$\gamma_S(x)=|\delta_{G}(\{x\},V(G)\setminus S)|$.

\begin{lemma} \label{justifs}
Let $G$ be a 3-edge-connected graph, $w\in\Bbb{Z}_{\geqslant 2}$, and let $S\subseteq V(G)$ be a set of vertices such that $|S|\geqslant w+1$ and $|\delta_{G}(S,V(G)\setminus S)|\leqslant 2w$.
If ${\bf tcw}(G)\leqslant w$, then $I(S,G)=(G[S],w,\partial_G(S),\gamma_S)$ is a {\sc Yes}-instance of {\sc Constrained Star-Cut Decomposition}.
\end{lemma}
\begin{proof}
 Let $(T,{\cal X})$ be a normalized tree-cut decomposition of $G$ of width at most $w$.
 We extend the weight function $\gamma_S$ on $\partial_G(S)$ into $\gamma'_S$ on $V(G)$ by setting $\gamma'_S(v)=\gamma_S(v)$ for
  every $v\in S$ and $\gamma'_S(v)=0$ otherwise. Also, given a subtree $T'$ of $T$, we let $\gamma_S'(T')=\sum_{t\in V(T')} \sum_{v\in X_t} \gamma_S'(v)$. The idea is to identify a node $t_c$ of $T$ that will serve as the central node of the star decomposition. The leaves of the star decomposition will results from the contraction of the subtrees of $T\setminus t_c$ containing bags that intersect the set $S$. To find the node $t_c$,  we orient the edges of $T$ using the following two rules. Given an edge $e=\{x,y\}\in E(T)$:
\begin{enumerate}\setlength\itemsep{.2em}
\item[] {\bf {\bf Rule 1}}: orient $e$ from $x$ to $y$ if $\gamma_S'(T_y)>w$.
\item[] {\bf {\bf Rule 2}}: orient $e$ from $x$ to $y$ if $S\cap \bigcup_{t\in V(T_x)}X_t=\emptyset$.
\end{enumerate}

Let $\vec{T}$ be the resulting orientation of $T$. Observe that Rule 1 and 2 may leave some edges of $T$ non-oriented.

\begin{claimN} \label{cl:orient}
For every edge $e=\{x,y\}$ of $T$, $e$ is oriented either in a single direction or not oriented in $\vec{T}$.
\end{claimN}
\begin{proofof}
Observe that if {\bf Rule 1} orients $e$ from $x$ to $y$, neither {\bf {\bf Rule 1}} nor {\bf {\bf Rule 2}} may orient $e$ in the opposite direction. The former is an immediate consequence of the fact $\gamma_S'(T_x)+\gamma_S'(T_y)=|\delta_{G}(S,V(G)\setminus S)|\leqslant 2w$. {\bf {\bf Rule 2}} does not orient $e$ from $y$ to $x$ either: if {\bf {\bf Rule 2}} does so, we have $S\cap \bigcup_{t\in V(T_y)}X_t=\emptyset$ and since the value $\gamma_S'(v)$ is non-zero only when $v\in S$, we conclude that $\gamma_S(T_y)=0$, a contradiction to the assumption that {\bf Rule 1} oriented $e$ from $x$ to $y$.  Moreover, the edge $e$ cannot be oriented in both directions by {\bf Rule 2} since $S$ is non-empty and thus at least one of the sets $\bigcup_{t\in V(T_x)}X_t$ and $\bigcup_{t\in V(T_y)}X_t$ intersects with $S$.
\end{proofof}

By Claim~\ref{cl:orient}, $\vec{T}$ contains at least one node, say $t_c$, which is not incident to an out-going edge in $\vec{T}$. Let $T_1,\ldots , T_{\ell}$ be the connected components of $T\setminus t_c$ containing a node $t$ such that $X_t\cap S\neq S$. Observe that as $|S|\geqslant w+1$ and $\tcw(T,\mathcal{X})=w$, $S$ cannot be included in a single bag of $(T,\mathcal{X})$ and thereby $\ell\geqslant1$. Consider the following tree-cut decomposition $(T^*,\mathcal{X}^*)$ of $G[S]$:
\begin{itemize}
\item[$\bullet$] $T^*$ is a star with central node $t_c$ and leaf nodes $t_1\dots t_{\ell}$,
\item[$\bullet$] the bag of node $t_c$ is $X^*_c=X_{t_c}\cap S$,
\item[$\bullet$] for every leaf node $t_i\in V(T^*)$, we set $X^*_i=\bigcup_{t\in V(T_i)} X_t\cap S$.
\end{itemize}

Observe that $(T^*,\mathcal{X}^*)$ is a tree-cut decomposition of $G[S]$ and since $|S|\geqslant w+1$, it is non-trivial.
By construction, as it is obtained from $(T,\mathcal{X})$ by contracting subtrees and removing  vertices from bags, we have that $\intcw(T^*,\mathcal{X}^*)\leqslant w$. It remains to prove that $|X_{t_c}|+\ell\leqslant w$ and that $\gamma_S(\partial_G(S)\cap X_t)\leqslant w$ for every leaf node $t$.
The former inequality directly follows from Observation~\ref{sldjflk34} and the fact that $(T,\mathcal{X})$ is an optimal tree-cut decomposition of $G$. The latter inequality follows from the fact that $t$ does not have an out-going edge in $\vec{T}$, in particular, {\bf Rule 1} does not orient any edge incident with $t$ outwardly from $t$.
\end{proof}

Given a $3$-edge-connected graph, applying Lemma~\ref{justifs} on a large leaf of a tree-cut decomposition that satisfies  the {\sl Invariant}, we obtain:

\begin{corollary} \label{cor:justifs}
Let $G$ be a 3-edge-connected graph $G$ such that $\mathbf{tcw}(G)\leqslant w$, and let
$t$ be a large leaf of a tree-cut decomposition $(T,\mathcal{X})$ satisfying    the {\sl Invariant}. Then $I(X_t,G)=(G[X_t],w,\partial_G(X_t),\gamma_{X_t})$ is a {\sc Yes}-instance of {\sc Constrained Star-Cut Decomposition}.
\end{corollary}

The next lemma shows that if  a large leaf bag $X_t$ of a tree-cut decomposition $(T,\mathcal{X})$ satisfying    the {\sl Invariant} defines a {\sc Yes}-instance of the \textbf{Constraint Tree-Cut Decomposition} problem, then $(T,\mathcal{X})$  can be further refined.

\begin{lemma} \label{lem:refine}
Let $G$ be a 3-edge-connected graph $G$ and $(T,\mathcal{X})$ be tree-cut decomposition of satisfying    the {\sl Invariant}. If $(T^*,\mathcal{X}^*)$ is a solution of  {\sc Constrained Star-Cut Decomposition} on the instance $I(X_t,G)=(G[X_t],w,\partial_G(X_t),\gamma_{X_t})$ where $t$ is a large leaf of $(T,\mathcal{X})$, then the pair
$(\tilde{T},\tilde{\mathcal{X}})$  where
\begin{itemize}
\item[$\bullet$] $V(\tilde{T})=(V(T)\setminus\{t\}) \cup V(T^*)$,
\item[$\bullet$] $E(\tilde{T})=(E(T)\setminus \{(t,t')\})\cup E(T^*)\cup \{(t_c,t')\}$, where $t'$ is the unique neighbor of $t$ in $T$ and $t_c$ is the central node of $T^*$,
\item[$\bullet$] $\tilde{\cal X}=({\cal X}\setminus \{X_{t}\})\cup{\cal X}^{*}$
\end{itemize}
is a tree-cut decomposition of $G$  satisfying    the {\sl Invariant}. Moreover the number of non-empty bags is strictly larger
in $(\tilde{T},\tilde{\mathcal{X}})$ than in $(T,\mathcal{X})$.
 \end{lemma}

\begin{proof}
By construction, $(\tilde{T},\tilde{\mathcal{X}})$ is a tree-cut decomposition of $G$. The fact that $(T^*,\mathcal{X}^*)$ is non-trivial implies that the number of non-empty bags is strictly larger in $(\tilde{T},\tilde{\mathcal{X}})$ than in $(T,\mathcal{X})$.

It remains to prove that $\intcw(\tilde{T},\tilde{\mathcal{X}})\leqslant 2\cdot w$. Since $(T^*,\mathcal{X}^*)$ is a solution to $I(X_t,G)$, we have $|X^*_{t_c}|+\ell\leqslant w$. As $G$
is edge $3$-connected, by Observation~\ref{sldjflk34}, the torso size at $t_c$ in $(\tilde{T},\tilde{\mathcal{X}})$ at most $w+1$, which is at most $2w$.
Let us verify that the adhesion of $(\tilde{T},\tilde{\mathcal{X}})$ is at most $2w$.
For this, it suffices to bound the value $|\delta^{\tilde{T}}(e)|$ for the newly created edges $e=\{t_i,t_c\}$, for all $i\in [\ell]$.
We have
\begin{align*}
|\delta^{\tilde{T}}(\{t_i,t_c\})|&=|\delta_{G}(\tilde{X}_{t_i},V(G)\setminus \tilde{X}_{t_i})|\\
& =|\delta_{G}(\tilde{X}_{t_i},X_t\setminus \tilde{X}_{t_i})|+|\delta_{G}(\tilde{X}_{t_i},V(G)\setminus X_t)|\leqslant 2w.\\
\end{align*}
The inequality follows from that $(T^*,\mathcal{X}^*)$ is a solution to $I(X_t,G)$. More precisely,
$|\delta_{G}(\tilde{X}_{t_i},X_t\setminus \tilde{X}_{t_i})|\leqslant w$ is implied by the fact that $\intcw(T^*;\mathcal{X})\leqslant w$. And $|\delta_{G}(\tilde{X}_{t_i},V(G)\setminus X_t)|\leqslant w$ is a consequence of  $\gamma_{X_t}(\partial_G(X_t)\cap X^*_{t_i})\leqslant w$.

Finally, as $(T^*,\mathcal{X}^*)$ is a non-trivial tree-cut decomposition, the number of non-trivial nodes is strictly larger in $(\tilde{T},\tilde{\mathcal{X}})$  than in $(T,\mathcal{X})$.
\end{proof}
\subsection{An FPT algorithm for {\sc Constrained Star-Cut Decomposition}}

Lemma~\ref{tfdwtcw} provides a quadratic bound on the treewidth of a graph in term of its tree-cut width. This allows us to develop a dynamic programming algorithm for solving \textsc{Constrained Star-Cut Decomposition} on graphs of bounded treewidth. To obtain a tree-decomposition, we use the 5-approximation {\sf FPT}-algorithm of the following proposition.


\begin{proposition}[see~\cite{BodlaenderDDFLP13anap}]\label{5approx}
There exists an algorithm which, given a graph $G$ and an integer $k$, either correctly decides that ${\bf tw}(G)>w$ or outputs a tree-decomposition of width at most $5w+4$ in time $2^{O(w)}\cdot n$.
\end{proposition}

If $\tcw(G)\leqslant w$, then by Lemma~\ref{tfdwtcw} $\tw(G)\leqslant 2w^2+3w$. From Proposition~\ref{5approx},
we may assume that $G$ has treewidth $O(w^2)$
and, based on this and the next lemma, solve
\textsc{Constrained Star-Cut Decomposition}
in  $2^{O(w^2\cdot \log w)}\cdot n$ steps.

A \emph{rooted tree decomposition} $(T,{\cal X},r)$ is a tree decomposition with a distinguished node $r$ selected as the \emph{root}. A \emph{nice tree decomposition} $(T,{\cal Y},r)$ (see~\cite{Klo94}) is a rooted tree decomposition where $T$ is binary, the bag at the root is $\emptyset$, and for each node $x$ with two children $y,z$ it holds $Y_x =Y_y =Y_z$, and for each node $x$ with one child $y$ it holds $Y_x =Y_y \cup \{u\}$ or $Y_x =Y_y \setminus \{ u\}$ for some $u \in V(G)$. Notice that a nice tree decomposition is always a rooted tree
decomposition. We need the following proposition.

\begin{proposition}[see~\cite{Bod96alin}]
\label{nice_tree}
For any constant $k\geqslant 1$, given a tree decomposition of a graph $G$ of
width $\leqslant k$ and $O(|V(G)|)$ nodes, there exists an algorithm that,
in $O(|V(G)|)$ time, constructs a nice tree decomposition of $G$ of width
$\leqslant k$ and with at most $4|V(G)|$ nodes.
\end{proposition}
\begin{lemma}
\label{ko5teod}
Let $(G,w,B,\gamma)$ be an input of \textsc{Constrained Star-Cut Decomposition}
and let $\tw(G)\leqslant q$.
There exists an algorithm that  given $(G,w,B,\gamma)$ outputs, if one exists,  a solution of $(G,w,B,\gamma)$
in $2^{O((q+w)\log w)}\cdot n$ steps.
\end{lemma}
\begin{proof-sketch}
From Proposition~\ref{nice_tree}, we can
assume that we are given  a nice tree-decomposition $(T,{\cal Y},r)$ of $G$ of width at most $O(q)$, which can be obtained in
time $2^{O(q)}\cdot n$ because of Proposition~\ref{5approx}. We describe dynamic programming tables. Let $Z_t$ be the vertex set $\bigcup_{t'\in V(T_t)} Y_{t'}$, where $T_t$ is the subtree of $T$ rooted at $t$. For every $1\leq \ell \leq w$, we
need to compute a collection $\mathcal{X}=\{X_0,\dots, X_{\ell}\}$ of pairwise disjoint subsets of $V(G)$ (some of them may be empty sets) such that $|X_0|+\ell\leqslant w$. The  subset $X_0$ will play the role of the bag of
the central node of the star-cut decomposition.

To guarantee that the specification of the problem can be checked, the dynamic programming table at node $t$ will store a collection of quadruples $(\phi,a,\alpha,\beta)$ with the following specifications:
\begin{itemize}
\item[(i)] $\phi: Y_t \rightarrow [0,\ell]$, indicating that $x\in Y_t$ belongs to $X_{\phi(x)}\cap Z_t$;
\item[(ii)] $a$ is a number indicating the size $|X_0\cap Z_t|$ of the central bag in $G[Z_t]$;
\item[(iii)] $\alpha: [\ell]\rightarrow [0,w]$ is a function, indicating the weight $\gamma(B\cap X_i\cap Z_t)$;
\item[(iv)] $\beta: [\ell] \rightarrow [0,w]$ is a function, indicating $|\delta(X_i,Z_t\setminus X_i)|$;
\end{itemize}
Suppose that we have constructed tables for all nodes of $T$ such that: for every node $t$, a quadruple $(\phi,a,\alpha,\beta)$ appears in the table at node $t$ if and only if there exists a collection $\mathcal{X}'=\{X'_0,\dots X'_{\ell}\}$ meeting the specifications. It is not difficult to see that the instance $(G,w,B,\gamma)$ is {\sc Yes} if and only if the table at the root contains a quadruple $(\phi,a,\alpha,\beta)$ such that $\ell +a \leq w$. Furthermore, such tables can be constructed using standard dynamic programming in a bottom-up manner.

Observe that the size of the dynamic table at each node $t$ is dominated by the number of collections $\mathcal{X}=\{X_0,\dots X_{\ell}\}$ of pairwise disjoint subsets of $Y_t$, with $\ell\leqslant w$, which is $2^{O((q+w)\log w)}$. Maintaining these tables follows by a standard dynamic programming algorithm.
%
\end{proof-sketch}

\subsection{Piecing everything together}

We now present a 2-approximation algorithm for {\sc Tree-cut Width} leading to the following result.

\begin{theorem}\label{90ncwlflsk}
There exists an algorithm that, given a graph $G$ and a $w\in\Bbb{Z}_{\geqslant 0}$, either outputs a tree-cut decomposition of $G$ with width at most $2w$
or correctly reports that no tree-cut decomposition of $G$ with width at most $w$ exists in $2^{O(w^{2}\cdot \log w)}\cdot n^{2}$ steps.
\end{theorem}
\begin{proof}
Recall that, by Lemmata~\ref{e92mfsd} and~\ref{i8i30eieue89}, we can assume that $G$ is $3$-edge-connected. If not, we iteratively decompose $G$ into 3-edge-connected components using the linear-time algorithm of~\cite{WatanabeTM93}. A tree-cut decomposition of $G$ can easily built from the tree-cut decomposition of its $3$-edge-connected components using Lemma~\ref{e92mfsd}. As mentioned earlier, the trivial tree-cut decomposition satisfies    the {\sl Invariant}. Let $(T,\mathcal{X})$ be a tree-cut decomposition satisfying    the {\sl Invariant}. As long as the current tree-cut decomposition $(T,\mathcal{X})$ contains a large leaf $\ell$, the algorithm applies the following steps repeatedly:
\begin{enumerate}
\item Let $X_{\ell}\in\mathcal{X}$ be the bag associated to a large leaf $\ell$. Compute a nice tree-decomposition of $G[X_{\ell}]$ of width at most $O(w^2)$ in $2^{O(w^{2})}\cdot n$ time. If such a decomposition does not exist, as $G[X_{\ell}]$ is a subgraph of $G$, Lemma~\ref{tfdwtcw} implies $\tcw(G)>w$ and the algorithm stops.
\item Solve \textsc{Constrained Star-Cut Decomposition} on $I(X_t,G)$ using the dynamic programming of Lemma~\ref{ko5teod} for $q = O(w^2)$ in time $2^{O(w^2\cdot \log w)}\cdot n$.
\item If $I(X_t,G)$ is a NO-instance, then by Corollary~\ref{cor:justifs}, $\tcw(G)>w$  and the algorithm stops.
\item Otherwise, by Lemma~\ref{lem:refine}, $(T,\mathcal{X})$ can be refined into a new tree-cut decomposition satisfying    the {\sl Invariant}.
\end{enumerate}

The algorithm either stops when we can correctly report that $\tcw(G)>w$ (step 1 or 3) or when the current tree-cut decomposition has no large leaf. In the latter case, as $(T,\cal{X})$ satisfies  (*), it holds that $\tcw(T,\mathcal{X})\leqslant 2\cdot w$. Observe that  each refinement step (step 4) strictly increases the number of non-empty bags (see Lemma~\ref{lem:refine}). It follows that the above steps are repeated at most $n$ times, implying that the running time of the 2-approximation algorithm is $2^{O(w^{2}\cdot \log w)}\cdot n^{2}$.
\end{proof}

\section{Open problems}

The main open question is on the possibility of improving the
running time or the approximation factor of our algorithm.
Notice that the parameter dependence $2^{O(w^{2}\cdot \log w)}$
is based on the fact that the tree-cut width is bounded by a quadratic
function of treewidth. As we proved (Theorem~\ref{theokg8h}),
there is no hope of  improving this upper bound. Therefore any
improvement  of the parametric dependence should avoid dynamic
programming on tree-decompositions or significantly improve the running time. Another issue is whether
we can improve the quadratic dependence on $n$ to a linear one.
In this direction we actually believe that an exact {\sf FPT}-algorithm for the tree-cut width can be constructed using the
``set of characteristic sequences'' technique, as this was done
for  other width parameters~\cite{ThilikosSB-1,ThilikosSB-2,BodlaenderK96,BodlaenderT04,Soares13Purs,JisuKO15cons}. However, as this technique is more involved, we believe that it would imply a
higher parametric dependence than the one of our algorithm.

\vspace{.35cm}

{\small\noindent\textbf{Acknowledgement}. We would like to thank the reviewers of the extended abstract of this work for helpful remarks that improved the presentation of the manuscript.}

\bibliographystyle{abbrv}
\bibliography{biblio-2px-tcw}

\end{document}